 \documentclass[final,3p,times]{elsarticle}
\usepackage{amssymb}
\usepackage{epstopdf}
\usepackage{graphicx}
\journal{*}

\usepackage{amsmath,amsthm}
\newcommand{\bra}[1]{\langle#1|}
\newcommand{\ket}[1]{|#1\rangle}
\usepackage{amsfonts}
\newtheorem{theorem}{Theorem}[section]

\newtheorem{proposition}[theorem]{Proposition}
\newtheorem{conjecture}[theorem]{Conjecture}

\theoremstyle{remark}

\theoremstyle{definition}

\theoremstyle{example}

\theoremstyle{notation}

\begin{document}

\begin{frontmatter}



\title{Paths of zeros of analytic functions describing finite quantum systems}


\author{H. Eissa}
\author{P. Evangelides}
\author{C. Lei}
\author{A. Vourdas}%\email{Corresponding author: A.Vourdas@Bradford.ac.uk}

\address{Department of Computing,\\
University of Bradford, \\
Bradford BD7 1DP, \\United Kingdom}

\begin{abstract}
Quantum systems with positions and momenta in ${\mathbb Z}(d)$, are described by the $d$ zeros of analytic functions on a torus.
The $d$ paths of these zeros on the torus, describe the time evolution of the system. 
A semi-analytic method for the calculation of these paths of the zeros, is discussed.  
Detailed analysis of the paths for periodic systems, is presented.
A periodic system which has the displacement operator to a real power $t$, as time evolution operator, is studied. 
Several numerical examples, which elucidate these ideas, are presented.
\end{abstract}

%\begin{keyword}
%%% keywords here, in the form: keyword \sep keyword
%Analytic representations \sep finite quantum systems \sep zeros of analytic functions
%%% PACS codes here, in the form: \PACS code \sep code
%\PACS 03.65.Aa\sep 02.30.Fn
%%% MSC codes here, in the form: \MSC code \sep code
%\MSC[2015] 81Q99\sep 30D99
%%% or \MSC[2008] code \sep code (2000 is the default)
%
%\end{keyword}
%
\end{frontmatter}



\section{Introduction}

There is an extensive literature on analytic representations in quantum mechanics, after the pioneering work by Bargmann \cite{B1,B2}.
The Bargmann analytic function in the complex plane studies problems related to the harmonic oscillator.
The zeros of the Bargmann function, which are also the zeros of the Husimi (or $Q$) function, provide a valuable insight to various quantum systems \cite{L1,L2,L3,L4,L5,L6,L7,L8}, chaos \cite{L8}, etc. Other potential applications include the study of two-dimensional electron gas in a magnetic field, quantum Hall effect, \cite{A1,A2,A3}, etc.

Analytic representations in the unit disc for problems with $SU(1,1)$ symmetry,
and analytic representations in the extended complex plane for systems with $SU(2)$ symmetry, have also been studied in the literature (reviews have been presented in \cite{P,H,V}).

Quantum systems with variables in ${\mathbb Z}(d)$ (the integers modulo $d$), have been studied extensively in the literature (e.g., \cite{F1,F2,F3,F4}).
Refs.\cite{ZV,TVZ,ELV,TVZ1, TVZ2} have represented their quantum
states with analytic functions on a torus, using Theta functions. It has been shown that these functions have exactly $d$ zeros, which determine uniquely the state of the system.
As the system evolves in time, the zeros follow $d$ paths, on the torus.
Ref \cite{L2} has also used a similar representation in studies of chaos.
Theta functions have been used extensively in various problems in physics \cite{Th1,Th2}.

In this paper we study different aspects of the zeros of analytic functions for finite quantum systems with variables in
${\mathbb Z}(d)$, as follows:
\begin{itemize}
\item
We propose in Eqs.(\ref{246}),(\ref{num}) a semi-analytic method for the calculation of the paths of the zeros, which is primarily analytical (section 2).  
Previous work is based on entirely numerical methods.
In principle the full quantum formalism can be expressed in terms of the $d$ zeros.
But it is difficult to express physical laws in terms of the zeros, without an analytical formalism that relates physical quantities with the zeros. 
The semi-analytical formalism in this paper, is a step in this direction.

\item
We study in detail the $d$ paths of the zeros of periodic systems. Each path is characterized by the multiplicity M, and by a pair of winding numbers $(w_1,w_2)$.
An interesting periodic system is one, which has as time evolution operator the displacement operator to a real power $t$. 
Displacement operators ${\cal Z}^{\alpha}{\cal X}^{\beta}$ are defined in finite quantum systems for $\alpha, \beta \in {\mathbb Z}(d)$, and it is interesting to
study these operators to a real power $t$. It is shown that the paths of the zeros are identical, but shifted with respect to each other (section 3).
\end{itemize}

\section{Analytic representation of finite quantum systems}\label{BB}

We consider a finite quantum system with variables in ${\mathbb Z}(d)$.
This system is described with the $d$-dimensional 
Hilbert space ${\cal H}(d)$.
Let $\ket{X; m}$ and $\ket{P;m}$ (where $m\in {\mathbb Z}(d)$) be the position and momentum bases  which  
are related through a Fourier transform, as follows:
\begin{alignat}{2}\label{Fou}
&\ket{P;n}={\cal F}|{X};n\rangle ;\qquad{\cal F}=d^{-1/2}\sum _{m,n}\omega (mn)\ket{X;m}\bra{X;n}
;\nonumber\\&\omega (m)=\exp \left [i \frac{2\pi m}{d}\right ]
\end{alignat}
Let $\ket {g}$ be an arbitrary state
\begin{eqnarray}\label{6}
&&\ket {g}=\sum _m g_m\ket {X;m};\qquad\sum _m |g_m|^2=1
\end{eqnarray}
We use the notation
\begin{eqnarray}\label{444}
&&\ket {g^*}=\sum _m g_m^*\ket {X;m};\;\;
\bra {g}=\sum _m g_m^*\bra {X;m}\nonumber\\&&\bra {g^*}=\sum _m g_m\bra {X;m}
\end{eqnarray}
We represent the state $\ket{g}$ with the analytic function \cite{L1,L2,ZV}
\begin{eqnarray}\label{aaa1}
G(z)=\pi^{-1/4} \sum_{m=0}^{d-1} g_m\;\Theta_3 \left [\frac{\pi m}{d}-z\sqrt{\frac{\pi}{2d}};\frac{i}{d}\right ]
\end{eqnarray}
where $\Theta_3$ is the Theta function \cite{T1}
\begin{eqnarray}\label{pa4}
&&\Theta _3(u,\tau)=\sum_{n=-\infty}^{\infty}\exp(i\pi \tau n^2+i2nu)\nonumber\\
&&\Theta _3'(u,\tau)=\frac{{\rm d}\Theta _3}{{\rm d}u}=i\sum_{n=-\infty}^{\infty}2n\exp(i\pi \tau n^2+i2nu).
\end{eqnarray}
We can prove that
\begin{eqnarray} \label{periodicity}
 &&G( z+\sqrt{2\pi d} ) = G(z)\nonumber\\
 &&G( z+ i \sqrt{2\pi d})  = G(z)\exp \left (\pi d-i z\sqrt{2\pi d} \right ),
\end{eqnarray}
and therefore it is sufficient to have this function in a cell 
\begin{alignat}{1}\label{hhh}
S=[M\sqrt{2\pi d},(M+1)\sqrt{2\pi d})\times [N\sqrt{2\pi d},(N+1)\sqrt{2\pi d})
\end{alignat}
where $(M,N)$ are integers labelling the cell.
Other models with more general quasi-periodic boundary conditions can also be studied.
The scalar product is given by 
\begin{alignat}{1}\label{scalar}
&\langle f^\ast| g \rangle =\frac{1}{d^{3/2}\sqrt{2\pi} }\int_S  {\rm d}\mu (z)F(z^*) G(z)=\sum f_mg_m;\nonumber\\
&{\rm d}\mu (z)={\rm d}^2z \exp \left( - z_I^2 \right)
\end{alignat}
These relations are proved using
the orthogonality relation\cite{TVZ}
\begin{alignat}{3}
 2^{-1/2}\pi ^{-1}d^{-3/2}\int_S {\rm d}\mu (z)&\Theta_3 \left [\frac{\pi n}{d}-z\sqrt{\frac{\pi}{2d}};\frac{i}{d}\right ]&\nonumber\\
&\times\Theta_3 \left [\frac{\pi m}{d}-z^*\sqrt{\frac{\pi}{2d}};\frac{i}{d}\right ]=\delta (m,n)
\end{alignat}
The coefficients $g_m$ in Eq.(\ref{6}) are given by 
\begin{alignat}{1}\label{5b}
g_m= 2^{-1/2}\pi ^{-3/4}d^{-3/2}\int_S {\rm d}\mu (z)
\Theta_3 \left [\frac{\pi m}{d}-z\sqrt{\frac{\pi}{2d}};\frac{i}{d}\right ]G(z^*).
\end{alignat}
It has been proved in \cite{L2,ZV}
that the analytic function $G(z)$ has exactly $d$ zeros $\zeta _n$ in each cell $S$, and that
\begin{eqnarray}\label{con}
\sum _{n =1}^d \zeta _n=\sqrt{2\pi d}(M+iN)+d^{3/2}\sqrt{\frac{\pi }{2}}(1+i).
\end{eqnarray} 
in finite systems the $d-1$ zeros define uniquely the state (the last zero is determined from Eq.(\ref{con})).
In infinite systems the zeros do not define uniquely the state.

If the $d-1$ zeros $\zeta _n$ are given, the last one can be found from Eq.(\ref{con}), and the function $G(z)$ is given by
\begin{alignat}{1}\label{500}
&G(z) = {\cal N}(\{\zeta _n\})\nonumber\\&\times\exp \left[ -i\sqrt{\frac{2\pi}{d}} Nz\right]
\prod_{n=1}^{d}  \Theta_3 \left[\sqrt{\frac{\pi}{2d}}(z-\zeta _n)+\frac{\pi (1+i)}{2};\; i\right]
\end{alignat}
Here $N$ is the integer that labels the cell (as in Eq.(\ref{hhh})), 
and ${\cal N}(\{\zeta _n\})$ is a normalization constant that does not depend on $z$ (see section 7 in ref\cite{ZV}).
Below we choose the cell with $M=N=0$.

\subsection{Time evolution and paths of zeros}
Let $H$ be the Hamiltonian of the system (a $d\times d$ Hermitian matrix $H_{mn}$).
As the system evolves in time $t$, each zero $\zeta _n$ follows a path $\zeta _n(t)$.

We consider infinitesimal changes to the coefficients from $g_m$ to $g_m+\Delta g_m$, where
\begin{eqnarray}\label{ppp}
\Delta g_m=i\Delta t \sum _nH_{mn}g_n
\end{eqnarray}
Then the zeros will change from 
$\zeta _n$ to $\zeta _n+ \Delta\zeta _n$.
From Eqs.(\ref{aaa1}),(\ref{500}) we get
\begin{alignat}{1}
&\pi^{-1/4} \sum_{m=0}^{d-1} (g_m+\Delta g_m)\;\Theta_3 \left [\frac{\pi m}{d}-z\sqrt{\frac{\pi}{2d}};\frac{i}{d}\right ]\nonumber\\
&=
{\cal N}(\{\zeta _k\})\prod_{n=1}^{d}  \Theta_3 \left[\sqrt{\frac{\pi}{2d}}(z-\zeta _n-\Delta \zeta _n)+\frac{\pi (1+i)}{2};\; i\right]
\end{alignat}
With a Taylor expansion of the right hand side, we get
\begin{alignat}{3}
&\pi^{-1/4} \sum_{m=0}^{d-1} \Delta g_m\;\Theta_3 \left [\frac{\pi m}{d}-z\sqrt{\frac{\pi}{2d}};\frac{i}{d}\right ]=
-{\cal N}(\{\zeta _k\})\sqrt{\frac{\pi}{2d}}\nonumber\\&\times \sum_{j=1}^{d}A_j(z)\Theta_3 '\left[\sqrt{\frac{\pi}{2d}}(z-\zeta _j)+\frac{\pi (1+i)}{2};\; i\right]\Delta \zeta _j\nonumber\\
&A_j(z)=\prod_{m\ne j}  \Theta_3 \left[\sqrt{\frac{\pi}{2d}}(z-\zeta _m)+\frac{\pi (1+i)}{2};\; i\right]
\end{alignat}
We insert $z=\zeta_n$ on both sides of this equation. For $j\ne n$ we get $A_j(\zeta _n)=0$. Therefore
\begin{eqnarray}
&&\pi^{-1/4} \sum_{m=0}^{d-1} \Delta g_m\;\Theta_3 \left [\frac{\pi m}{d}-\zeta _n\sqrt{\frac{\pi}{2d}};\frac{i}{d}\right ]\nonumber\\&&=
-{\cal N}(\{\zeta _k\})\sqrt{\frac{\pi}{2d}}A_n(\zeta _n)\Theta_3 '\left[\frac{\pi (1+i)}{2};\; i\right]\Delta \zeta _n\nonumber\\
&&A_n(\zeta _n)=\prod_{m\ne n}  \Theta_3 \left[\sqrt{\frac{\pi}{2d}}(\zeta _n-\zeta _m)+\frac{\pi (1+i)}{2};\; i\right]
\end{eqnarray}
Using Eq.(\ref{pa4}), we found numerically that
\begin{eqnarray}
\Theta_3 '\left[\frac{\pi (1+i)}{2};\; i\right]=1.9888i.
\end{eqnarray}
Therefore we have analytical expressions for the derivatives of the functions $\zeta _n(g_0,...,g_{d-1})$:
\begin{eqnarray}\label{246}
&&\frac{\partial \zeta _n}{\partial g_m}=-\frac{\pi^{-1/4} \Theta_3 \left [\frac{\pi m}{d}-\zeta _n\sqrt{\frac{\pi}{2d}};\frac{i}{d}\right ]}{{\cal N}(\{\zeta _k\})\sqrt{\frac{\pi}{2d}}A_n(\zeta _n)\Theta_3 '\left[\frac{\pi (1+i)}{2};\; i\right]}
\end{eqnarray}
We use them for numerical calculations as
\begin{alignat}{1}\label{num}
\zeta _n+\Delta \zeta _n=\zeta _n+\sum _m \frac{\partial \zeta _n}{\partial g_m}\Delta g_m=\zeta _n+i\Delta t\sum _{m,k} \frac{\partial \zeta _n}{\partial g_m} H_{mk}g_k.
\end{alignat}
In each step of the iteration process ${\cal N}(\{\zeta _k\})$ is calculated as
\begin{eqnarray}
{\cal N}(\{\zeta _k\})=\frac{\pi^{-1/4} \sum_{m=0}^{d-1} g_m\;\Theta_3 \left [\frac{\pi m}{d}-z\sqrt{\frac{\pi}{2d}};\frac{i}{d}\right ]}{\prod_{n=1}^{d}  \Theta_3 \left[\sqrt{\frac{\pi}{2d}}(z-\zeta _n)+\frac{\pi (1+i)}{2};\; i\right]}
\end{eqnarray}
and is used in the next step. As we mentioned earlier, the ${\cal N}(\{\zeta _k\})$ does not depend on $z$ and any value of $z$ can be used for its numerical calculation.
Since $\sum _m|g_m|^2=1$ the $\Delta g_m$ are subject to the constraint
\begin{eqnarray}
\sum _m[g_m^*\Delta g_m+g_m(\Delta g_m)^*]=0.
\end{eqnarray}
Ref\cite{TVZ} calculated the paths of the zeros indirectly, using a computationally expensive approach.
It calculated the vector $\exp(itH)\ket {f}$ and then the analytic function $f(z;t)$, at each time $t$.
Then a MATLAB function was used to find the zeros, at each time $t$.

In all calculations of the present paper, we go directly from the zero $\zeta _n(t)$ to the zero $\zeta _n(t+\Delta t)$ with Eqs.(\ref{246}),(\ref{num}).
For the calculation of ${\cal N}(z;\{\zeta _n\})$ we need the coefficients $g_m$ at each step.
At $t=0$ we start from given values of zeros, which we insert in Eq.(\ref{aaa1}) and get a system of $d$ equations with $d$ unknowns.
This gives the coefficients $g_m$ at $t=0$ (which we normalize). At later times the $g_m$ becomes $g_m+\Delta g_m$, where $\Delta g_m$ is the step used in Eqs.(\ref{ppp}),(\ref{num}). 
The present method is semi-analytic and therefore computationally less expensive and more accurate.

\section{Periodic systems}
We consider periodic systems with Hamiltonians such that $\exp(iTH)={\bf 1}e^{i\theta}$ for some $T$, and some phase factor $e^{i\theta}$ (which does not change the physical state).
This occurs when the ratios of the eigenvalues of $H$ are rational numbers.

Results analogous to those in sections \ref{S1}-\ref{S3}, have been reported in ref\cite{TVZ}, using an entirely numerical method.
Our results here are based on the semi-analytical method described in section \ref{BB}.

\subsection{Multiplicity $M$ of paths of zeros:}\label{S1}
We consider the Hamiltonian
\begin{equation} 
H=
\begin{bmatrix}\label{e1}
1&1&0&0\\
1&1&0&0\\
0&0&1&0\\
0&0&0&1
\end{bmatrix}
\end{equation}
In this case the period is $T=2\pi$.
We assume that at $t=0$ the zeros are the following:
\begin{alignat}{2}\label{z11}
&\zeta_0(0)=1-1.99i ;&\quad&\zeta_1(0)= 3.02+3i;\quad\nonumber\\
&\zeta_2(0) =1+3i;&\quad&\zeta_3(0)=-0.01+1i.
\end{alignat}
Using Eq.(\ref{num}) we have calculated the paths of the zeros.
Results are shown in Fig.\ref{f1} (see also Fig.4 in ref\cite{TVZ}).
In this case there are two paths with multiplicity $M=1$ and one path with multiplicity $M=2$.
For clarity, the figures show regions which might be larger or smaller than one cell (which is a square with each side equal to
$\sqrt{2\pi d}$). They also show the position of the zeros at various times.
We note that
\begin{alignat}{2}
&\zeta_0(T)=\zeta_3(0) ;&\quad&\zeta_3(T)=\zeta_0(0)\quad\nonumber\\
&\zeta_1(0) =\zeta_1(T);&\quad&\zeta_2(0)=\zeta_2(T).
\end{alignat}
It is seen that although the set of zeros has period $T$, a particular zero (e.g., $\zeta_0$ or $\zeta_3$) returns to its original position after time which is a multiple of $T$ 
(in this example $2T$).

We also consider the Hamiltonian of Eq.(\ref{e1}) and assume that at $t=0$ the zeros are 
\begin{alignat}{2}\label{z2b}
&\zeta_0(0)=2-2.99i;&\quad&\zeta_1(0)= 2.02-2.01i;\quad\nonumber\\
&\zeta_2(0) =1-1.01i;&\quad&\zeta_3(0) =-0.01+i;\quad
\end{alignat}
in which case the results are shown in Fig.\ref{f2}.
In this case we have one path with multiplicity $M=4$. Here
\begin{alignat}{2}
&\zeta_0(T)=\zeta_2(0) ;&\quad&\zeta_2(T)=\zeta_3(0)\quad\nonumber\\
&\zeta_3(T) =\zeta_1(0);&\quad&\zeta_1(T)=\zeta_0(0).
\end{alignat}

\subsection{Winding numbers $(w_1,w_2)$ of paths of zeros:}\label{S2}
We consider the Hamiltonian
\begin{equation} \label{eq3}
H=
\begin{bmatrix}
1.5&0.2&0\\
0.2 &1.5 &0\\
0&0 &2.1\\
\end{bmatrix}
\end{equation} 
In this case the period is $T=5\pi$. 
We assume that at $t=0$ the zeros are the following:
\begin{alignat}{1}\label{z4}
&\zeta_0(0)=1.01+2i;\;
\zeta_1(0)=2.15+2.56i\nonumber\\
&\zeta_2(0)=3.35+1.95i.
\end{alignat}
The paths of zeros are shown in Fig.\ref{f3}.
The winding numbers $(w_1,w_2)$ of the three paths are $(0,0)$, $(0,1)$ and $(0,1)$.

\subsection{Joining of two paths of zeros into a single path:}\label{S3}
We consider the Hamiltonian
\begin{equation} 
H=
\begin{bmatrix}\label{eg}
0&1&0&0&0\\
1&0&0&0&0\\
0&0&1&0&1\\
0&0&0&1&1\\
0&0&1&1&0\\
\end{bmatrix}
\end{equation}
In this case the period is $T=2\pi$.
We consider two cases where at $t=0$ the zeros are given by
\begin{alignat}{2}\label{za1}
&\zeta_0(0)=1.3+4.16i;&\quad&\zeta_1(0)=0.12+2.03i;\nonumber\\
&\zeta_2(0)=0.11+1.62i;&\quad&\zeta_3(0)=-2.6-0.2i;\nonumber\\
&\zeta_4(0)=-1.71+.81i,&&
\end{alignat}
and also by
\begin{alignat}{2}\label{za2}
&\zeta_0(0)=1.3+4.16i;&\quad&\zeta_1(0)=0.1+2.0i;\nonumber\\
&\zeta_2(0)=0.11+1.62i;&\quad&\zeta_3(0)=3-0.2i;\nonumber\\
&\zeta_4(0)=-1.69+0.83i.&&
\end{alignat}
The paths of zeros in these two examples are shown in Figs.\ref{f4},\ref{f5}, correspondingly. 
In these figures we see how by changing the initial zeros slightly, two paths (highlighted)with multiplicity $1$, join together into one path (highlighted) with multiplicity $2$.

\subsection{Zeros of the analytic representation of ${\cal X}^t\ket{g}$}\label{S4}
Displacement operators in the ${\mathbb Z}(d)\times {\mathbb Z}(d)$ phase space, are defined as
\begin{eqnarray}\label{99}
&&{\cal Z}=\sum _{n}\omega (n)|{X};n\rangle \langle {X};n|\nonumber\\
&&{\cal X}=\sum _{n}\omega (-n)|{P};n\rangle \langle {P};n|\nonumber\\
&&{\cal Z}\ket{{P};n}=\ket{{P};n+1};\;\;\;\;\;{\cal X}\ket{{ X};n}=\ket{{X};n+1}\nonumber\\
&&{\cal X}^{d}={\cal Z}^{d}={\bf 1};\;\;\;\;\;\;
{\cal X}^\beta {\cal Z}^\alpha = {\cal Z}^\alpha {\cal X}^\beta \omega (-\alpha \beta);\;\;\;\;\;\;
\end{eqnarray}
where $\alpha, \beta \in {\mathbb Z}(d)$.
We study the zeros of the analytic representation of ${\cal X}^t\ket{g}$ which we denote as  $G(z;t)$,
where $t\in {\mathbb R}$.
${\cal X}^t$ can be viewed as a time evolution operator $\exp(itH)$ with Hamiltonian
$H=-i\ln {\cal X}$ (the logarithm is multi-valued and we take the principal value).

Let $\ket{g}=\sum {\widetilde g}_m\ket{P;m}$, where ${\widetilde g}_m$ are the Fourier transforms of the $g_m$ in Eq.(\ref{6}). 
The state ${\cal X}^t\ket{g}=\sum [\omega (-m)]^t{\widetilde g}_m\ket{P;m}$ is represented by the function
\begin{alignat}{1}
G(z;t)=\pi ^{-1/4}\exp\left(-\frac{z^2}{2}\right)\sum_{m=0}^{d-1}\exp\left(-\frac{i2\pi mt}{d}\right ){\widetilde g}_m&\nonumber\\
\times\Theta_3\left[\frac{\pi m}{d}-iz\sqrt{\frac{\pi}{2d}};\frac{i}{d}\right]&
\end{alignat}
We used here the fact that the $\ket {P;m}$ is represented by $\pi ^{-1/4}\exp\left(-\frac{z^2}{2}\right)\Theta_3\left[\frac{\pi m}{d}-iz\sqrt{\frac{\pi}{2d}};\frac{i}{d}\right]$\cite{ELV}.

Let $\zeta _n (t)$ where $n=0,...,d-1$ be the zeros of $G(z;t)$, i.e.,
\begin{alignat}{1}
G[\zeta _n (t);t]=0
\end{alignat}
The index $n$ labels the various paths of zeros.
\begin{proposition}\label{propo}
Each path of zeros $\zeta _{n+\beta}$ of $G(z;t)$ (that represents ${\cal X}^t\ket{g}$), is a shifted version (in the real direction) of another path $\zeta _{n}$.
The position of the zero on each path at a certain time, is the same as the position of the zero on another path, at a different time:
\begin{alignat}{1}\label{3b}
\zeta _{n+\beta} (t+\beta)=\zeta _n (t)+\beta \sqrt {\frac{2\pi}{d}};\;\;\;\;n,\beta\in {\mathbb Z}(d). 
\end{alignat}
\end{proposition}
\begin{proof}
We assume that $G[\zeta _n (t);t]=0$ and prove  $G[\zeta _{n+\beta} (t+\beta);t+\beta]=0$, where the `new path' $\zeta _{n+\beta} (t+\beta)$ is given in Eq.(\ref{3b}).
We express the Theta function as a sum as in Eq.(\ref{pa4}) and change the summation from $n\in {\mathbb Z}$ into $k_n=n-t\in {\mathbb Z}-t$.
We get
\begin{eqnarray}
G(z;t)
&=&\pi^{-1/4}\exp\left(-\frac{z^2}{2}\right)\exp\left(tz\sqrt{\frac{2\pi}{d}}-\frac{\pi t^2}{d}\right)\nonumber\\&\times &\sum_{m=0}^{d-1}{\widetilde g}_m\sum_{k_n}\exp\left(\frac{2i\pi mk_n}{d}\right)\nonumber\\
&\times&\exp\left(2k_nz\sqrt{\frac{\pi}{2d}}-\frac{\pi k_n^2}{d}-2\frac{\pi k_nt}{d}\right).
\end{eqnarray}
We insert $z=\zeta _{n+\beta} (t+\beta)$ in $G(z;t)$ and change the variable $t'=t-\beta$. Using $G[\zeta _n (t);t]=0$ we prove that $G[\zeta _{n+\beta} (t+\beta);t+\beta]=0$.
\end{proof}
In fig\ref{f6} we plot the paths of the zeros of the state $X^t\ket{g}$. The state $\ket{g}$ is defined through the zeros at $t=0$ which are
\begin{alignat}{1}\label{z7}
&\zeta_0(0)=1.54+2.47i;\;\zeta_1(0)=2.01+2.18i\nonumber\\&\zeta_2(0)=2.95+1.86i
\end{alignat}
It is seen that there are $d$ identical paths, which are shifted in the $z_R$-direction by $\sqrt{2\pi/d}$.
We note that at a particulat time the zeros do not obey the relation $\zeta _{i+1}(t)=\zeta _{i}(t)+\sqrt{2\pi/d}$ (e.g., the zeros at $t=0$ which are shown in fig\ref{f7}).
However, the whole path of a zero over a period, is a shifted version of the path of another zero. 

General displacement operators in the ${\mathbb Z}(d)\times {\mathbb Z}(d)$ phase space, are defined as ${\cal D}(\alpha, \beta)={\cal Z}^{\alpha}{\cal X}^{\beta}\omega (-2^{-1/2})$.
In this part of the paper we assume that $d$ is an odd integer, so that $2^{-1}$ exists in ${\mathbb Z}(d)$.
Let $e_m$ and $\ket{u_m}$ the eigenvalues and eigenvectors of ${\cal D}(\alpha, \beta)$.
We consider the state $\ket{g}=\sum r_m\ket{u_m}$. 
We assume that the  state $[{\cal D}(\alpha, \beta)]^t\ket{g}=\sum e_m^t\ket{u_m}$ is represented by the function ${\mathfrak G}(z,t)$.
Let $\zeta _n (t)$ where $n=0,...,d-1$ be the zeros of ${\mathfrak G}(z;t)$, i.e.,
\begin{alignat}{1}
{\mathfrak G}[\zeta _n (t);t]=0
\end{alignat}
We give the following conjecture, which is a generalization of proposition \ref{propo} for general displacement operators.
\begin{conjecture}
Each path of zeros of ${\mathfrak G}(z,t)$ (that represents $[{\cal D}(\alpha, \beta)]^t\ket{g}$), is a shifted version (in both the real and imaginary 
direction) of another path.
\end{conjecture}
This conjecture is supported with the numerical result in Fig\ref{f7},
where we plot the paths of the zeros of ${\mathfrak G}(z,t)$ which represents the state $[{\cal D}(1,1)]^t\ket{g}$.
The state $\ket{g}$ is defined through the zeros at $t=0$, which are
\begin{alignat}{1}\label{z100}
&\zeta_0(0)=1.4-2.01i;\;
\zeta_1(0)=2.15+2.32i\nonumber\\&
\zeta_2(0)=-1.39+1.86i.
\end{alignat}

\section{Discussion}
We have considered quantum systems with positions and momenta in ${\mathbb Z}(d)$.
An analytic representation on a torus that uses Theta functions, which describes these systems, has been given in Eq.(\ref{aaa1}).
The $d$ zeros of these analytic functions define uniquely the state of the system.
As the system evolves in time the zeros follow $d$ paths on the torus.

A semi-analytic method for the calculation of these paths of the zeros, has been given in Eqs.(\ref{246}),(\ref{num}).  
It has been used for the study of the paths of periodic systems.
Each path is characterized by the multiplicity M, and by a pair of winding numbers $(w_1,w_2)$.
Other phenomena like the joining of two paths of zeros into a single path, have also been studied.
The case that the time evolution operator is the displacement operator to a real power $t$, has also been studied (section \ref{S4}).
In this case the paths of the zeros are identical, but shifted with respect to each other.

There are deep links between the zeros of analytic functions and the behaviour of quantum systems.
For systems with finite dimensional Hilbert space, the zeros determine the state of the system,
and the time evolution can be described with $d$ classical paths on a torus.
The ultimate goal is to develop the full quantum formalism in terms of the zeros, and to derive general laws that describe their motion.
For example, it is interesting to study what determines the velocity and acceleration of the zeros.
Analytical relations between the zeros and the various quantum quantities, would be ideal for this purpose. 
The semi-analytical method proposed in this paper, is a positive step in this direction.

Other related problems, like the behaviour of the zeros in the semiclassical limit, could also be studied in extensions of the present work. 

\begin{figure}[h]
\centering
\includegraphics[width=0.5\textwidth]{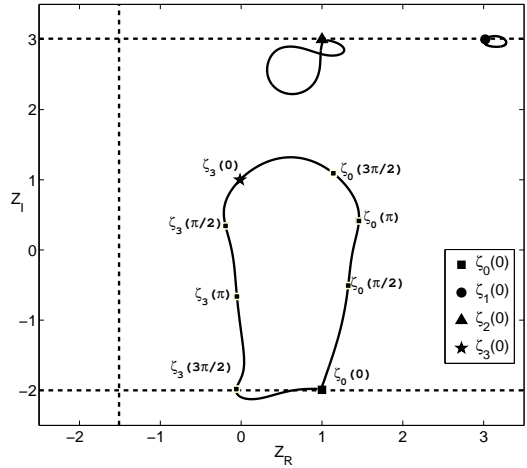}
\caption{Paths of the zeros for the Hamiltonian of Eq.(\ref{e1}). 
The period is $T=2\pi$.
At $t=0$ the zeros are given in Eq.(\ref{z11}).
Dotted lines show a cell, which is a square with each side equal to $5.01$}.
\label{f1}
\end{figure}
\begin{figure}[h]
\centering
\includegraphics[width=0.5\textwidth]{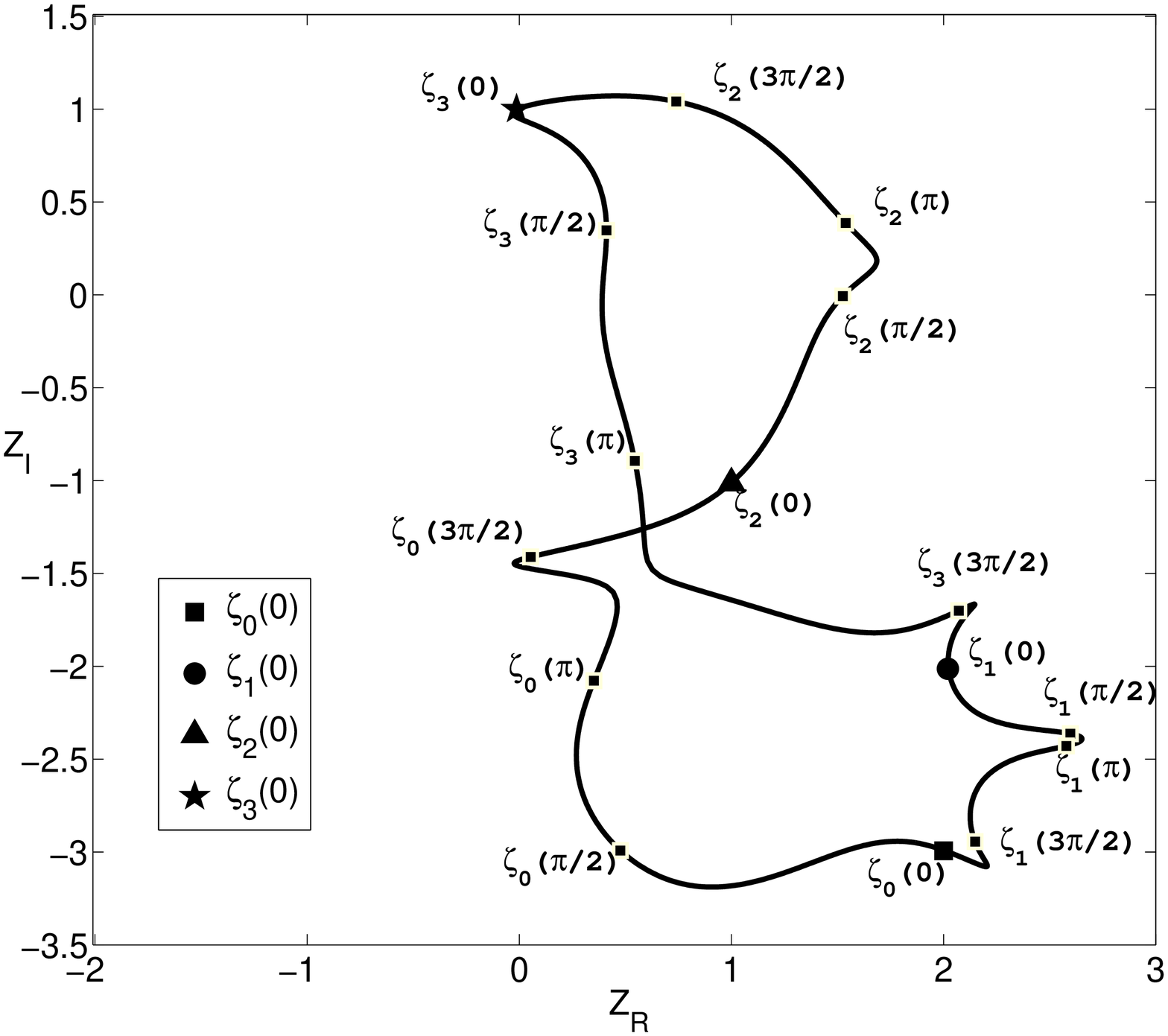}
\caption{Paths of the zeros for the Hamiltonian of Eq.(\ref{e1}). 
The period is $T=2\pi$.
The zeros at $t=0$ are given in Eq.(\ref{z2b}).The cell is a square with each side equal to $5.01$}
\label{f2}
\end{figure}
\begin{figure}[h]
\centering
\includegraphics[width=0.5\textwidth]{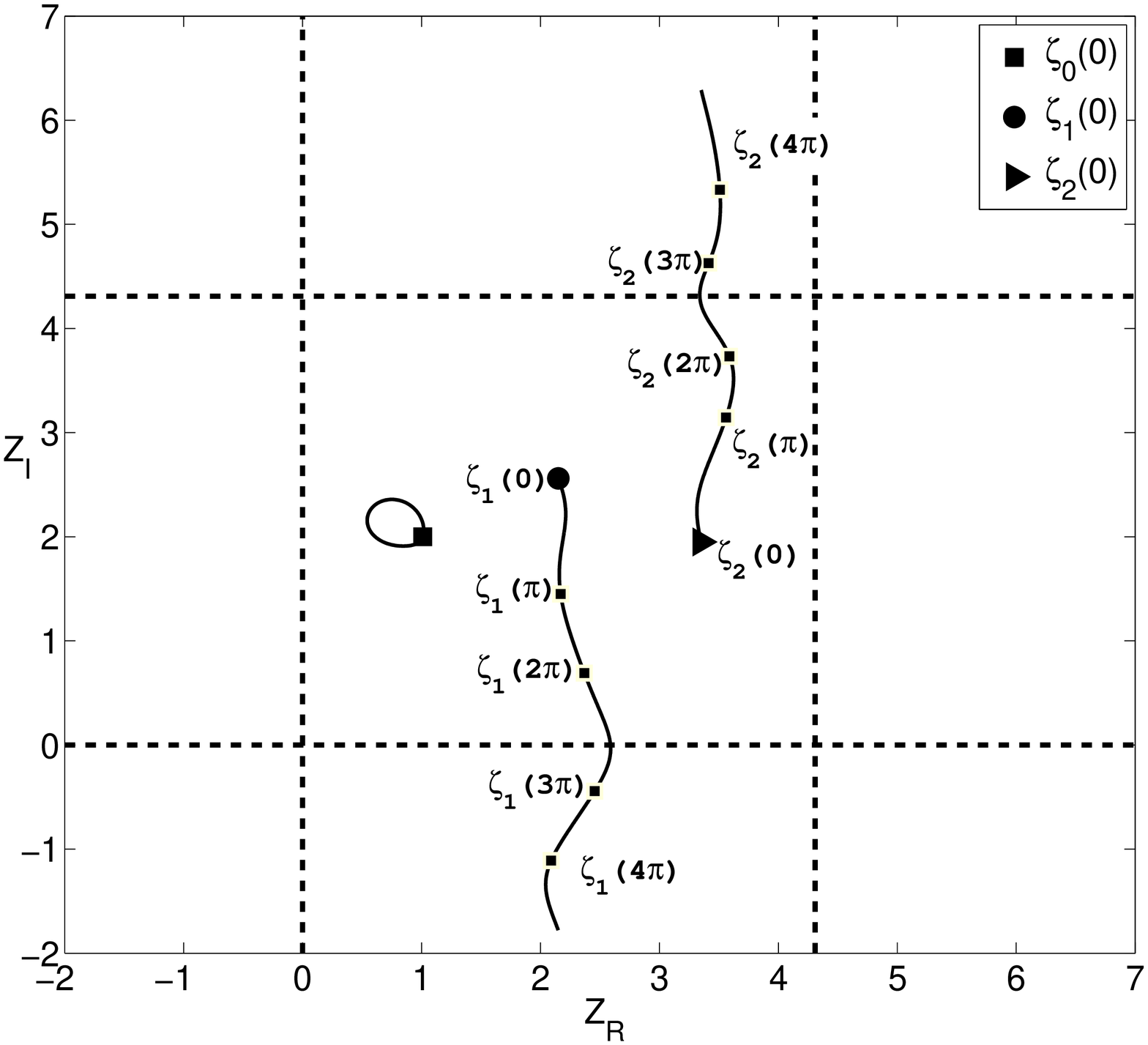}
\caption{Paths of the zeros for the  Hamiltonian of Eq.(\ref{eq3}). 
The period is $T=5\pi$.
The zeros at $t=0$ are given in Eq.(\ref{z4}).
Dotted lines show a cell, which is a square with each side equal to $4.34$}
\label{f3}
\end{figure}
\begin{figure}[h]
\centering
\includegraphics[width=0.5\textwidth]{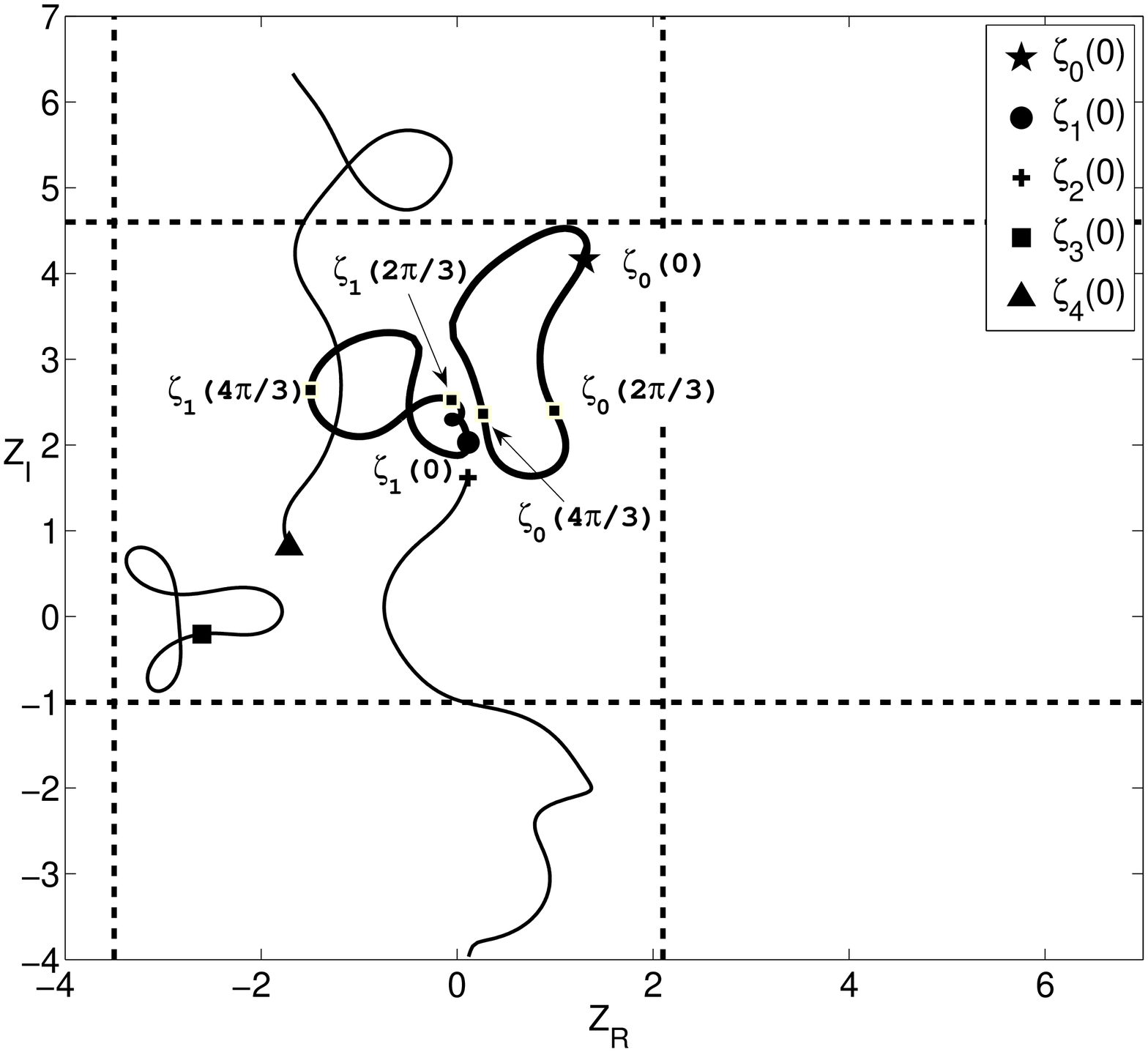}
\caption{Paths of the zeros for the Hamiltonian of Eq.(\ref{eg}). 
The period is $T=2\pi$.
At $t=0$ the zeros are given in Eq.(\ref{za1}).
Dotted lines show a cell, which is a square with each side equal to $5.6$}
\label{f4}
\end{figure}
\begin{figure}[h]
\centering
\includegraphics[width=0.5\textwidth]{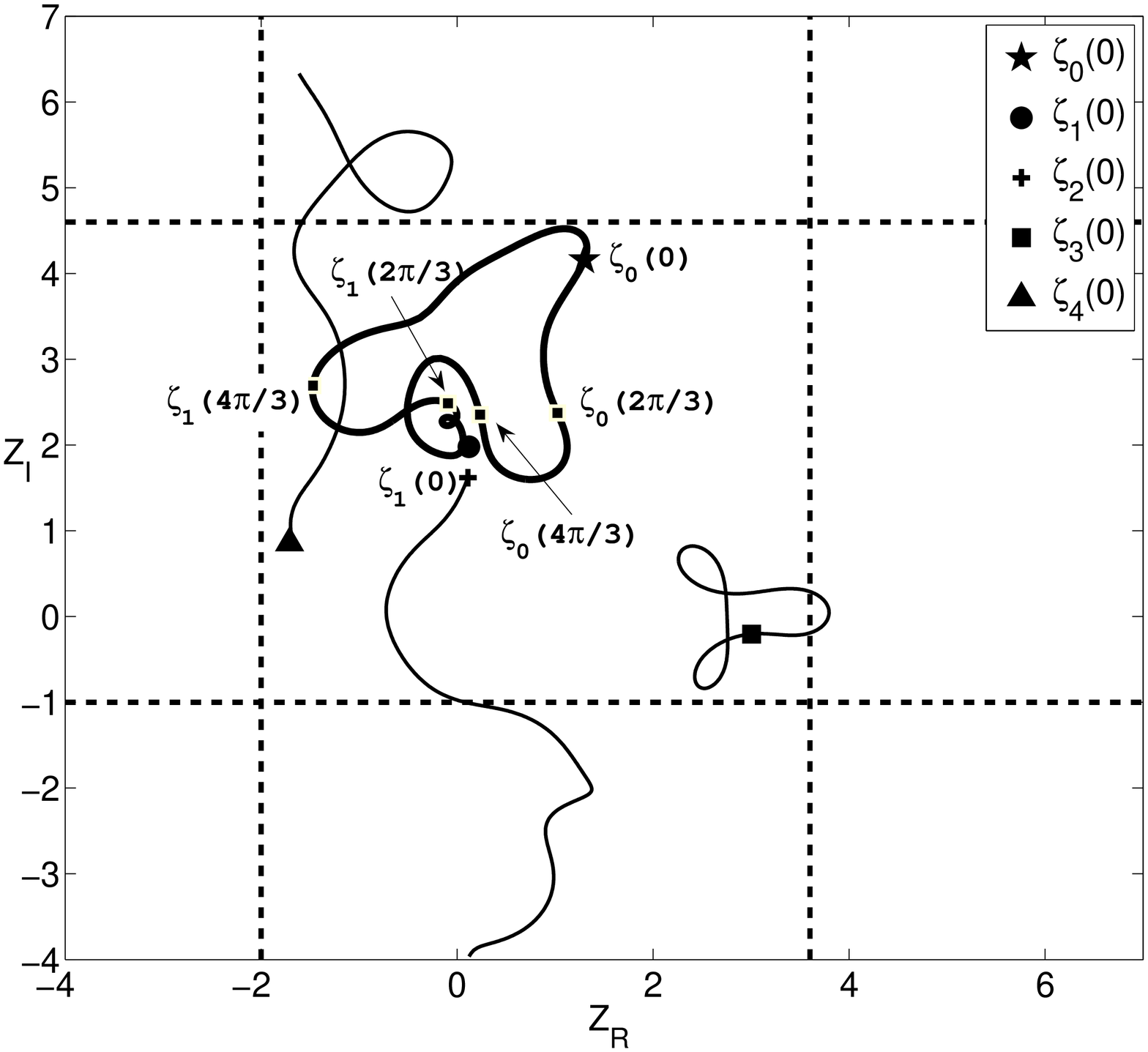}
\caption{Paths of the zeros for the Hamiltonian of Eq.(\ref{eg}). 
The period is $T=2\pi$.
At $t=0$ the zeros are given in Eq.(\ref{za2}).
Dotted lines show a cell, which is a square with each side equal to $5.6$}
\label{f5}
\end{figure}
\begin{figure}[h]
\centering
\includegraphics[width=0.5\textwidth]{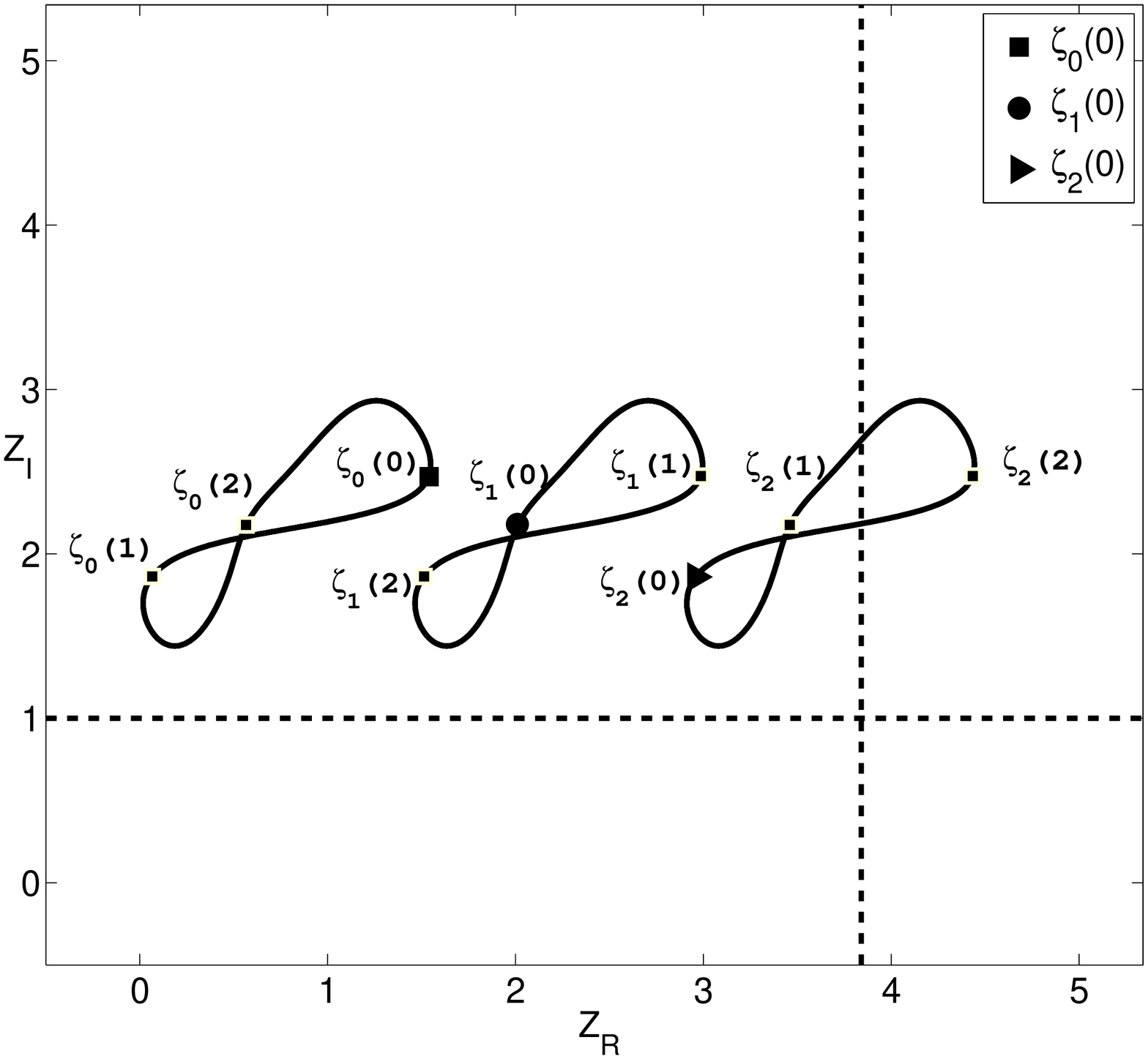}
\caption{Paths of the zeros of the analytic representation of the state ${\cal X}^t\ket{g}$. 
The period is $T=3$.
The state $\ket{g}$ is defined through the zeros at $t=0$ given in Eq.(\ref{z7}).
Dotted lines show a cell, which is a square with each side equal to $4.34$}
\label{f6}
\end{figure}
\begin{figure}[h]
\centering
\includegraphics[width=0.5\textwidth]{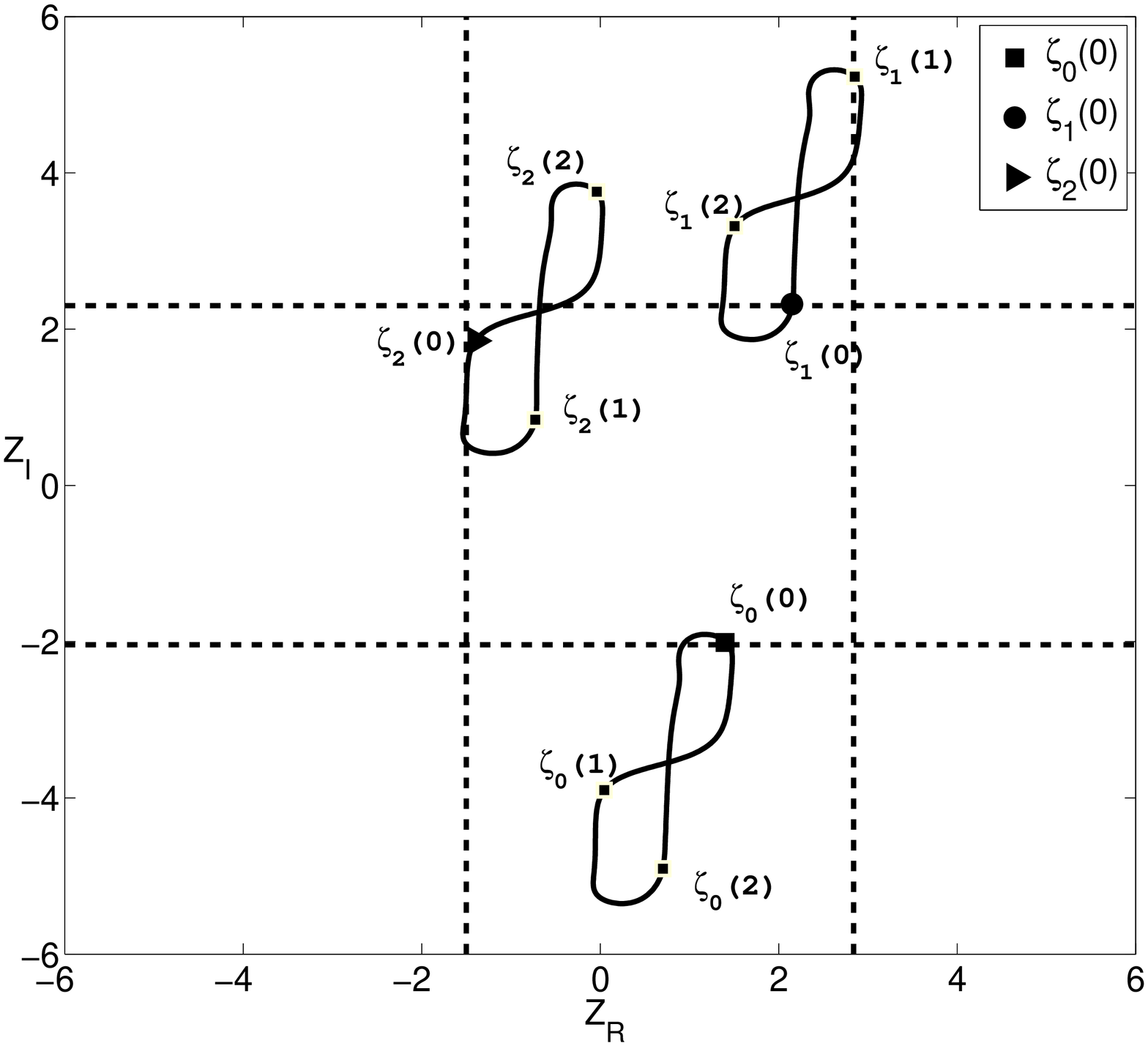}
\caption{Paths of the zeros of the analytic representation of the state $[{\cal D}(1,1 )]^t\ket{g}$.
The period is $T=3$.
The state $\ket{g}$ is defined through the zeros at $t=0$ given in Eq.(\ref{z100}).
Dotted lines show a cell, which is a square with each side equal to $4.34$}
\label{f7}
\end{figure}
\end{document}